\documentclass[12pt]{article}
\usepackage{amsmath,amsfonts,amssymb,amsthm}
\usepackage{mathtools, tikz-cd, graphicx, comment, hyperref}
\usepackage[margin=1in]{geometry}
\usepackage[parfill]{parskip}
\usepackage[all]{xy}
\usepackage[normalem]{ulem}

\newcounter{counter}
\counterwithin{counter}{subsubsection}

\theoremstyle{definition}
\newtheorem{dfn}[counter]{Definition}

\newtheorem{rmk}[counter]{Remark}

\theoremstyle{plain}
\newtheorem{thm}{Theorem}

\newtheorem{lem}[counter]{Lemma}

\newtheorem{cor}[counter]{Corollary}

\newcommand{\End}{\mathop{\mathrm{End}}}
\newcommand{\Sym}{\mathop{\mathrm{Sym}_2}}
\newcommand{\Symp}{\mathop{\mathrm{Sym}_2^+}}

\newcommand{\id}{\mathop{\mathbf{id}}\nolimits}

\newcommand{\ActSym}
{\mathop{\mathrm{Act}_{\mathrm{Sym}}}}
\newcommand{\ActSec}
{\mathop{\mathrm{Act}_{\mathrm{Sec}}}}
\newcommand{\ActFun}
{\mathop{\mathrm{Act}_{\mathrm{Fun}}}}

\newcommand{\An}{\mathop{\mathit{An}}}

\title{Convolutions and Gaussians in Renormalization}
\author{Raymond Puzio,
        Sam McCrosson}

\begin{document}
\maketitle

\begin{abstract}
The Kadanoff-Wilson-Fisher approach to renormalization is based upon studying the renormalization transform, which may be described as an action of the monoid \(\mathbb{R}^\times_{\ge 1}\) on a suitable space of interactions. It is typically computed by manipulating the path integral or the perturbation series.  Here we will present an alternative algebraic description of the renormalization transform. We treat the space of interactions as a semigroup under convolution and act on it with a Lie group associated with the quantum harmonic oscillator. 
\end{abstract}

\tableofcontents

\pagebreak

\section{Introduction} \label{Intro}

Renormalization started as a computational trick that physicists used to get finite answers to various superficially-divergent integrals describing physical quantities (electron charge, mass, etc.).\cite{huang2013critical}  Later on, Stueckelberg and Petermann \cite{stueckelberg1953normalization} and Gell-Mann and Low \cite{gell1954quantum} reinterpreted renormalization in terms of a semigroup of transformations on the space of coupling constants.  Extending this interpretation, Kadanoff, Wilson, and Fisher \cite{fisher1998renormalization} developed a framework for describing how interactions among systems change as the systems are analyzed at different scales.

Over time, the mathematical formalism of renormalization has been developed and refined.  Bogoliubov and collaborators produced a recursive procedure for computing renormalization via counterterms \cite{bogoliubov1960introduction, hepp1966proof, zimmermann2000convergence}.
Later on, Epstein and Glazer produced another such procedure. \cite{epstein1973role}  The mathematics underlying these procedures has since been reformulated in terms of
Hopf algebras \cite{ConnesKriemer1999}.
Proceeding in a different direction, Polchinski produced a differential equation which describes the continuous version of Wilsonian renormalization \cite{polchinski1984renormalization}.

In this paper, we present a novel perspective on renormalization that centers on the algebraic properties of Gaussian distributions under convolution. 
Specifically, in \autoref{Intro}, we give a brief overview of renormalization in QFT -- the reader familiar with the topic may want to skim this section for notation.
In \autoref{Alg}, we develop a Lie group $\mathsf{Osc}(V)$ and use it to enumerate the properties of Gaussian distributions that will be useful for describing renormalization.
In \autoref{sec:renorm}, we present renormalization as an action of the (Lie) monoid $\mathbb{R}^\times_{\geq 0}$ on $C^\infty(V)$ that factors through a one-parameter submonoid of $\mathrm{GL}(V)\ltimes \Symp(V)$, culminating in \autoref{eq:RenormSummary}.

To keep the exposition at an elementary level, we make several simplifying assumptions.  We only treat bosonic field theory in Euclidean signature, and we only prove results for finite-dimensional integrals. We'll treat the infinite-dimensional space of field configurations relevant to perturbation theory in a future work.

\subsection{Informal overview of quantum field theory}

We begin by reviewing some basic quantum field theory. 

Given a spacetime manifold $M$ (generally, $M=\mathbb{R}^4$), \textit{fields} are quantities defined at every point on $M$.  The space of all possible field values on $M$ are typically organized as a bundle $E\to M$.  A particular field configuration is a section of this bundle. The space of all field configurations will be denoted $V$.  In most examples from physics, the fibers are all vector spaces, so \(V\) is also a vector space. 

Some field configurations are more likely to occur than others. This is typically quantified by an \textit{action functional} $S: V \to \mathbb{R}$. We'll assume $S$ can be expressed as a sum of two terms, $S(\phi) = \phi \mathbf{Q}\phi + I(\phi)$ where $\mathbf{Q}\in V^*\otimes V^*$ describes the kinetic / free part of the theory, and $I(\phi)\in C^\infty(V)$ describes the potential / interacting part of the theory. We'll further assume $\mathbf{Q}$ has an inverse quadratic form $\mathbf{P}\in V\otimes V$ (in the sense of \autoref{eq:inverseForms}), which we'll call the \emph{propagator} of the theory.

For our purposes, a \emph{(continuum) (quantum field) theory} will be specified by its propagator $\mathbf{P}$ and interaction term $I$, which jointly specify its action.

\textit{Observables} are numerical quantities (real or complex) we can associate to field configurations. Formally, an observable is an integrable function $\mathcal{O}:V \to \mathbb{C}$.
For example, given a point $x\in M$, we can define an observable by evaluation: 
\begin{equation}
    \mathcal{O}_x(\phi) = \phi(x)
\end{equation}
If $M$ is the earth and the field is temperature, $\mathcal{O}_x$ just reports the temperature at $x$ of the field configuration $\phi$. 

Classically, Hamilton's \textit{principle of least action} posits that nature conspires to achieve whichever field configuration $\phi$ minimizes the action, which can be found computationally by setting the derivative of $S$ to 0. In the quantum setting, we entertain all possible field configurations, not just the one that minimizes the action. Instead, field configurations $\phi$ are weighted by a probability $\exp(-S(\phi))$, and if we want to know the expected value of an observable $\mathcal{O}$ we integrate its values over all field configurations, weighted by their probabilities:
\begin{equation}
    \langle \mathcal{O} \rangle
    = \frac{\int_{V} \mathcal{O}(\boldsymbol\phi) \exp \{-S(\boldsymbol\phi)\} \, d\mu(\boldsymbol\phi)}{\int_{V}  \exp \{-S(\boldsymbol\phi)
    \} \, d\mu(\boldsymbol\phi)}
    .
\end{equation}
The other fundamental quantities of interest are \textit{correlations} of observables. In brief, given two observables $\mathcal{O}_1$ and $\mathcal{O}_2$, we want to know how likely it is that they are correlated (related to each other). In a classical setting, this notion is redundant because there's a single optimal field configuration $\phi$ and the value of any two observables will be evaluated on the same field configuration, which means they are ``totally correlated."
In a quantum setting, we compute expectation values $\langle \mathcal{O}_1\rangle $ and $\langle \mathcal{O}_2\rangle$ and wonder how much the two expectations are drawing from the same field configurations. The easiest way to do this is to compute the $L^2$ inner product of observables. Indeed, we typically define correlation values to be this inner product:
\begin{equation}
    \langle \mathcal{O}_1,\mathcal{O}_2\rangle := \int_V \mathcal{O}_1(\phi)\mathcal{O}_2(\phi)\, \exp(-S(\phi))\,d\phi
\end{equation}
From this, if we want a probabilistic interpretation of how correlated the observables are (as a number between 0 and 1), we can use the standard formula for the correlation coefficient of two random variables in terms of moments (c.f. \cite{stats} p.265):
\begin{equation}
    \mathsf{corr}(\mathcal{O}_1,\mathcal{O}_2) = \frac{\langle \mathcal{O}_1,\mathcal{O}_2\rangle - \langle \mathcal{O}_1\rangle\langle \mathcal{O}_2\rangle}{\sqrt{\langle \mathcal{O}_1^2\rangle - \langle \mathcal{O}_1\rangle^2}\cdot\sqrt{\langle \mathcal{O}_2^2\rangle - \langle \mathcal{O}_2\rangle^2}}
\end{equation}
When the observables are given by evaluating the field at different points $x,y\in M$, the correlation $\langle O_x,\mathcal{O}_y\rangle$ tracks how correlated expected field values are at the different points $x$ and $y$. An important question in quantum field theory and statistical mechanics is determining how \(\mathsf{corr}(\mathcal{O}_x,\mathcal{O}_y)\) behaves as \(|x - y| \to \infty\).

However, when one proceeds to evaluate the expectation value for observables such as two-point correlations, one runs into difficulties.  If one takes \(I\) to be a typical non-trivial interaction term such as \(I(\boldsymbol{\phi}) = \int_M \phi (\boldsymbol{x})^4 \, d\boldsymbol{x}\), the integral over \(V\) diverges.  This difficulty can be traced to the fact that 
there are field configurations $\phi\in V$ that vary rapidly on short length scales, resulting in large values of $S(\phi)$ that cause the expectation value to diverge.

To remedy this misbehavior, physicists introduced the techniques of regularization and renormalization.  Regularization is the operation of modifying a physical theory so as to suppress degrees of freedom at short length scales. Renormalization is a technique for studying how a physical theory behaves above different length-scale cutoffs by assigning to each length scale an effective theory which describes physics above that length scale, and relating the effective theories above different length scales. Combining these ideas, one can replace a poorly behaved continuum theory by a family of effective theories indexed by length scale. These effective theories typically all have the same form but differ in the values of parameters. Furthermore, when one computes the values of observables in these theories, the results converge to sensible limits as the length scale approaches zero as opposed to diverging.

There are multiple ways to regularize a theory.  The one which we will use is known as the ``smooth momentum cut-off'' (c.f. \cite{costello2011renormalization} p.42 or \cite{huang2010quantum} p.412) and consists of replacing the propagator \(\mathbf{P}\) with a propagator \(\mathbf{P}_{L}\) in which the contribution of length scale values smaller than \(L\) has been damped them by means of a factor which rapidly decreases to zero as \(L\) approaches zero. 
This class includes several popular choices, such as Pauli-Villars regularization and heat kernel regularization.  For instance, in flat Euclidean space, the heat kernel regularized propagator is given as
\begin{equation}
  P_{L} (\mathbf{x}, \mathbf{y})
  = \int_{L}^{\infty}  (4 \pi l)^{-n/2}
  \exp\left(- \frac{|\mathbf{x} - \mathbf{y}|^2}{4l} \right) \, dl
\end{equation}
An important feature of both Pauli-Villars regularization and heat kernel regularization is that they are equivariant under dilation and monotonic.  By monotonicity, we mean that \(\mathbf{P}_{x} - \mathbf{P}_{y}\) is a positive quadratic form when \(x < y\).  For the heat kernel regularization, this property follows immediately from the fact that the integrand is positive.  The effect of dilation on the heat kernel regularization works out to be
\begin{align}
   P_{c L} (\mathbf{x}, \mathbf{y})
   &= \int_{c L}^{\infty}  (4 \pi l)^{-\tfrac{n}{2}} \,
  \exp\left(- \frac{|\mathbf{x} - \mathbf{y}|^2}{4l} \right) \, dl \\
  &= c^{1-\tfrac{n}{2}} \int_{L}^{\infty} \,
  (4 \pi l)^{-\tfrac{n}{2}} \,\exp\left(- \frac{|\mathbf{x}
      - \mathbf{y}|^2}{4c l} \right) \, dl \\
    &= c^{\tfrac{2-n}{2}} P_{L} (\mathbf{x}c^{-1/2}, \mathbf{y}c^{-1/2})\label{eq:HKequivar}
\end{align}

Henceforth, we shall assume that our cut-off is equivariant under dilation and monotonic.  In algebraic terms, we may describe this equivariance in terms of an action of the multiplicative group of reals on the vector space of fields $V$, i.e. 
\begin{align}
    \mathbf{T}:\mathbb{R}^\times_{>0} \times V &\to V\\
    \mathbf{T}_c(\mathbf{v}) := \mathbf{T}(c,\mathbf{v}) &= c^\frac{2-n}{4}\mathbf{v}(c^{-1/2}\mathbf{x})
\end{align}
which we can use to define an action on linear operators, \((\mathbf{T}\otimes\mathbf{T}) \colon \mathbb{R}^{\times}_{>0} \times \Symp (V) \to  \Symp (V)\).  In terms of this action, we may re-express the conclusion of \autoref{eq:HKequivar} as \(\mathbf{P}_{c L} = (\mathbf{T}_{c}\otimes \mathbf{T}_c) \mathbf{P}_{L}\).  

Once a theory is regularized to suppress degrees of freedom below length-scale $L_0$, one can renormalize to produce an effective theory for length scale \(cL_0\) by carrying out a three step procedure (c.f. \cite{huang2010quantum} p.320):

\begin{enumerate} \label{threeSteps}
    \item ``Make blocks" : Separate out the degrees of freedom at a length scale less than \(cL_0\) from the degrees of freedom at a length scale larger than \(cL_0\).
    \item ``Coarse grain" : Average away all features at length scales smaller than \(cL_0\).  
    \item Rescale the theory from \(cL_0\) to a fiducial length scale \(L_0\)
\end{enumerate}
The calculations for implementing these three steps with a smooth cut-off are well known and can be found in multiple textbooks.  Before presenting our version, we need to define some notation.

\subsection{Notation}\label{sec:notation}

To set up some notation, let $\Sym (V) \subset V \otimes V$ denote the vector space of symmetric 2-tensors on $V$, and let $\Symp (V)$ denote the positive-definite symmetric 2-tensors on $V$. 
To keep our calculations concise, we will make use of a variant of matrix multiplication notation for tensors.  
Given a tensor \(\mathbf{C} \in V \otimes V = \sum_i u_i \otimes v_i\) and a dual vector \(\mathbf{k}\), we can contract them to obtain dual vector
\begin{align}
    \mathbf{Ck} &= \sum_i u_i \otimes \mathbf{k} (v_i) \\
    \mathbf{kC} &= \sum_i \mathbf{k}(u_i) \otimes v_i.
\end{align}
Note that, if \(\mathbf{C}\) is symmetric, \(\mathbf{kC} = \mathbf{Ck}\).  We can contract twice to obtain a scalar,
\begin{equation}
    \mathbf{kCk} = \sum_i \mathbf{k} (u_i) \otimes \mathbf{k} (v_i).
\end{equation}
Recall, $\mathsf{End}(V) \cong V \otimes V^*$, so we can define the inverse tensor \(\mathbf{C}^{-1} = \sum_i \mathbf{h}_i \otimes \mathbf{k}_i \in V^* \otimes V^*\) by the requirement that 
\(\mathbf{C} \mathbf{C}^{-1} = \id \in \mathsf{End}(V)\) where
\begin{equation} \label{eq:inverseForms}
    \mathbf{C} \mathbf{C}^{-1} = \sum_i \mathbf{h}_i (\mathbf{u}_i) \, \mathbf{v}_i \otimes \mathbf{k}_i.
\end{equation}
In general, we can describe this notation as follows:
Suppose that \(A,B,C\) are vector spaces and that \(\mathbf{M} = \sum_i a_i \otimes b'_i \in A \otimes B^*\) and \(\mathbf{N} = \sum_j b_j \otimes c_j \in B \otimes C\).  We shall define \(\mathbf{MN} = \sum_{ij} b'_i (b_j) \, a_i \otimes c_j  \in A \otimes C\).  Making use of the isomorphisms \(V \cong V^{**}\) and \(V \cong \mathbb{R} \otimes V \cong V\otimes \mathbb{R}\), we can extend this notation. 

\subsection{Algebraic reformulation}

In this paper, we will show how the renormalization procedure described above can be described algebraically. 
For this algebraic reformulation, all we'll need is a vector space $V$, a quadratic form $\mathbf{P}\in\Symp (V)$, a smooth function $I\in C^\infty(V)$, and an action $\mathbf{T} : \mathbb{R}^\times_{>0} \times V\to V$.
To keep this exposition elementary, we will assume $V$ is finite-dimensional throughout.
In \autoref{Alg}, we will introduce a group called the \emph{oscillator group}, $\mathsf{Osc}(V)$, constructed from $\mathrm{GL}(V)$ and $V$ using semidirect products and study how it acts on $C^\infty(V)$ treated as a ring with multiplication given by convolution. In \autoref{sec:renorm}, we will use these algebraic facts to construct both the discrete renormalization transform, $\mathcal{R}_c: C^\infty(V)\to C^\infty(V)$, and its continuous analog. 
Before delving into the details, we begin with a broad outline.

As a computational device for describing the renomalization process, we'll use something called the \emph{generating function} of a theory. The exponential of the action $\exp(-S(\phi))$ is a sort of probability distribution on the space of field configurations $V$. By taking the Laplace transform of $\exp(-S(\phi))$, we can produce a function $Z:V^*\to\mathbb{R}$ that can be used to recover the correlations of any two observables:\\

\begin{dfn} \label{def:generatingFunction}
We define the \emph{moment generating function} $Z[\mathbf{P},I]:V^*\to\mathbb{R}$ of a theory with propagator $\mathbf{P}$ and interaction term $I$ to be:
\begin{equation}
    Z[\mathbf{P},I](\mathbf{J}) = \int_{V}
  \exp (-\tfrac{1}{2} \boldsymbol\phi \mathbf{Q} \boldsymbol\phi +
  I (\boldsymbol\phi) + \mathbf{J}(\boldsymbol\phi))
  \, d\boldsymbol\phi
\end{equation}
Sometimes it is advantageous to take a logarithm to define a related function, the \emph{cumulant generating function}:
\begin{equation}
    W[\mathbf{P},I] = \log(Z[\mathbf{P},I])
\end{equation}
\end{dfn}

\begin{rmk}
    Expanding in powers of $J$, we can recover the moments (c.f. \cite{zee2010quantum}, p.50):
\begin{equation}
    \frac{Z(\mathbf[\mathbf{P},I]({J})}{Z[\mathbf{P},I](0)} = \left\langle \exp \{\mathbf{J}(\boldsymbol\phi)\} \right\rangle =
    \sum_{n=1}^\infty \frac{1}{n!} \mathbf{J}^{\otimes n} \left( \left\langle \boldsymbol\phi^{\otimes n} \right\rangle \right)
\end{equation}

More generally, we can recover the expectation value of any observable $\mathcal{O}$ from the generating function as
\begin{equation}\label{eq:ObsZ}
  \langle \mathcal{O} \rangle = \frac{1}{Z[\mathbf{P},I](0)}
  \left[ \mathcal{O}
    \left( \frac{\partial}{\partial\mathbf{J}} \right)
    Z[\mathbf{P},I](\mathbf{J}) \right]_{\mathbf{J} = 0}.
\end{equation}
\end{rmk}

Note that if we take \(\mathcal{O}\) to be a monomial in \(\phi\), this formula extracts the corresponding term in the Taylor series of \(Z[\mathbf{P},I]\). Also observe that we can rescale $Z[\mathbf{P},I](\mathbf{J})$ by any constant without changing the calculation of the expectation above.

We begin \autoref{sec:renorm} with the observation that we can re-express the integral defining the generating function as a convolution in \autoref{GenAsConvo}:
\begin{equation}
    W[\mathbf{P},I](\mathbf{J}) = \tfrac{1}{2} \mathbf{JPJ} + \tilde{W}[\mathbf{P},I](\mathbf{PJ})
\end{equation}
where we define
\begin{equation}
  \tilde{W}[\mathbf{P}, I](\mathbf{X}) =
  \log ((\mathcal{N}(\mathbf{P}) * \{\exp \circ I\})(\mathbf{X})).
\end{equation}

Using the action of $\mathsf{Osc}(V)$ on $C^\infty(V)$ developed in \autoref{Alg}, we can implement the three step renormalization procedure using facts about convolutions and Gaussian distributions.
Naively, we'd want to describe a ``continuum theory" whose associated length-scale is 0. For technical reasons, this is impossible, so we assume a nonzero fiducial length-scale $L_0$:

\begin{enumerate}
    \item We ``make blocks" by determining which energy scale $cL_0$ we're renormalizing to by writing \( P_{L_0} = (P_{L_0} - P_{cL_0}) + P_{cL_0} \).  
    \item We``coarse grain" by using \autoref{eq:PIadj} to redo the interaction term as $I_{\text{eff}} = \tilde{W}[P_{L_0}-P_{cL_0},I]$, which can be used in the generating function:
    \begin{equation}
        \tilde{W}[P_{L_0},I](\mathbf{J}) = \tilde{W}[P_{c L_0},I_{\text{eff}}](\mathbf{J})
    \end{equation}
    \item We finish by rescaling using \autoref{lem:rescaling} and the transformation $\mathbf{T}_c:V\to V$ to get
    \begin{equation}
        W[P_{c L_0},I_{\text{eff}}](\mathbf{J}) = W[P_{L_0},\mathcal{R}_{c} I](\mathbf{J}\circ \mathbf{T}_c)
    \end{equation}
\end{enumerate}

We can combine all three of these steps to 
describe renormalization as an action of a suitable one-parameter submonoid of \(\mathsf{GL} (V) \ltimes (\Symp (V), +)\) on \(C^{\infty} (V)\).  
 
The facts we needed are well-known and can be proven using elementary algebra and calculus.  However, we will rederive them using abstract algebra because that approach gives a deeper understanding and makes possible generalizations to situations where the elementary methods do not apply.  
 
\section{Algebraic Facts} \label{Alg}

\subsection{The Oscillator Group} \label{sec:oscgroup}

In this section, we construct the \emph{oscillator group} \(\mathsf{Osc} (V)\).
This group, along with some of its subgroups, will play a central role throughout this section for characterizing properties of Gaussians, which will be subequently used in \autoref{sec:renorm} to describe the process of renormalization.
The construction will involve taking multiple semidirect products of familiar groups from linear algebra. This can be done in more than one way, so we shall perform the operations in an order which is suited to the applications which we shall present later.

\subsubsection{Constructing $\mathsf{Osc}(V)$}

Let \(V\) be a real vector space (finite dimensional for now).   Associated to this vector space, there are several Lie groups:
\begin{enumerate}
    \item The general linear group \(\mathsf{GL}(V) \subset \End(V)\) which consists of all invertible linear transforms of \(V\).
    \item The additive group \((V,+)\) of the vector space.
    \item The additive group \((V^*, +)\) of the dual vector space.    
\end{enumerate}

Starting with these groups, we will now construct more groups as semidirect products.
Let $\mathbb{R}^+$ denote the additive group of reals. Given a real vector space \(V\), we define the (Lie) groups \(\mathsf{ILin}(V), \mathsf{IGL}(V), \mathsf{Heis}(V), \mathsf{Osc}(V)\) as follows:\\

\begin{dfn}
    Define the \emph{Inhomogeneous Linear Group} \(\mathsf{ILin}(V) = V \times \mathbb{R}^+\).  Note that the product of groups is the instance of the semidirect product using the trivial group action.\\
\end{dfn}

\begin{dfn}
    Define the \emph{Inhomogeneous General Linear Group} \(\mathsf{IGL}(V) = \mathsf{GL}(V) \ltimes V^*\) where the semidirect product is formed using the right action
    \begin{align}
        V^* \times \mathsf{GL}(V) &\to V^* \\
        (\mathbf{k},\mathbf{M}) &\mapsto \mathbf{kM}.
    \end{align}
    where $\mathbf{kM}$ is the dual vector defined by composition:
    \begin{equation} \label{dualmatrixcomposite}
        \begin{tikzcd}
            V \ar[r, "\mathbf{M}"] \ar[rr, bend right, "\mathbf{kM}"'] & V \ar[r, "\mathbf{k}"] & \mathbb{R}
        \end{tikzcd}
    \end{equation}  
    Explicitly, the group operation in $\mathsf{IGL}(V)$ can be written:
    \begin{equation}
        (\mathbf{L}, \mathbf{j}) \bullet  (\mathbf{M}, \mathbf{k})
        = (\mathbf{LM}, \mathbf{jM} + \mathbf{k}).
    \end{equation}
\end{dfn}

\begin{rmk}
    Note that if we think of a matrix $\mathbf{M}\in\mathsf{GL}(V)\subset V\otimes V^*$, then the dual vector $\mathbf{kM}$ above is consistent with the contraction of tensors in \autoref{sec:notation}.\\
\end{rmk}

\begin{dfn}
    Define the \emph{Heisenberg Group} \(\mathsf{Heis}(V) = V^* \ltimes \mathsf{ILin}(V)\) where the semidirect product is formed using the left action
    \begin{align}
        V^*\times (V\times\mathbb{R}^+) &\to V\times\mathbb{R}^+\\
        (\mathbf{k},\mathbf{v},a) &\mapsto (\mathbf{v},a+\mathbf{k}(\mathbf{v}))
    \end{align}
    Explicitly, the group operation is then:
    \begin{equation}
        (\mathbf{j},\mathbf{u},a)\bullet(\mathbf{k},\mathbf{v},b) = (\mathbf{j}+\mathbf{k},\mathbf{u}+\mathbf{v},a+b+\mathbf{j}(\mathbf{v}))
    \end{equation}
\end{dfn}

\begin{dfn} 
    Define the \emph{Oscillator Group} \(\mathsf{Osc}(V) = \mathsf{IGL}(V) \ltimes \mathsf{ILin}(V)\) where the semidirect product is formed using the left action
    \begin{align}
        (\mathsf{GL}(V) \ltimes V^*) \times (V \times \mathbb{R}) &\to V \times \mathbb{R} \\
        (\mathbf{M}, \mathbf{k}), (\mathbf{v}, c) &\mapsto (\mathbf{Mv}, c + \mathbf{k} (\mathbf{v}))
    \end{align}
    The group operation in $\mathsf{Osc}(V)$ is then
    \begin{equation}\label{eq:OscGpOper}
        (\mathbf{L}, \mathbf{j}, \mathbf{u}, a) 
        \bullet (\mathbf{M}, \mathbf{k}, \mathbf{v}, b)
        = (\mathbf{LM}, \mathbf{jM} + \mathbf{k}, \mathbf{u} + \mathbf{Lv}, a + b + \mathbf{j} (\mathbf{v})).
    \end{equation}
\end{dfn}

\begin{rmk} \label{subgroups}
    Note that $\mathsf{Heis}(V)$ can be witnessed as a subgroup of $\mathsf{Osc}(V)$ with the identity matrix in the first slot:
    \begin{align}
        \mathsf{Heis}(V) &\hookrightarrow \mathsf{Osc}(V)\\
        (\mathbf{k},\mathbf{v},a) &\mapsto (\id,\mathbf{k},\mathbf{v},a)
    \end{align}
    We can also construct the oscillator group as a semidirect product \(\mathsf{Osc}(V) = \mathsf{GL}(V) \ltimes \mathsf{Heis}(V)\).  In what follows, we will make use of both of these ways of factoring \(\mathsf{Osc}(V)\) as a semidirect product.
    Using this observation, we can write the following lattice of subgroups:
    \begin{equation}
        \xymatrix{
            && \mathsf{Osc}(V) \ar@{-}[ld] \ar@{-}[rd] &&&& \\
            & \mathsf{IGL}(V) \ar@{-}[ld] \ar@{-}[rd] && \mathsf{Heis}(V) \ar@{-}[ld] \ar@{-}[rd] &&& \\
            \mathsf{GL}(V) && V^* && \mathsf{ILin}(V) \ar@{-}[ld] \ar@{-}[rd] & \\
            & & & V & & \mathbb{R}^+
            }
    \end{equation}
\end{rmk}

\begin{rmk}
Intuitively, we can think of an element \((\mathbf{v}, c) \in \mathsf{ILin}(V)\) as an inhomogeneous linear function
\begin{align}
    V &\to \mathbb{R} \\
    \mathbf{x} &\mapsto c\mathbf{x} + \mathbf{v}
\end{align}
and the group operation is pointwise addition of functions. 

We can think of an element \((\mathbf{M}, \mathbf{k}) \in \mathsf{IGL}(V)\) as a transform
\begin{align}
    V^* &\to V^* \\
    \mathbf{y} &\mapsto \mathbf{y}\mathbf{M} + \mathbf{k}
\end{align}
and the group operation is composition of transforms.\\
\end{rmk}

\begin{rmk}
 Note that the group \(\mathsf{Osc}(V)\) has the same underlying space as the group \(\mathsf{GL}(V \oplus \mathbb{R})\).  However, their multiplication operations are not isomorphic.  To see the difference,
 we may note that, once we choose a basis for \(V\), \(\mathsf{Osc}(V)\) has a faithful matrix representation
 \begin{equation}
     (\mathbf{M}, \mathbf{k}, \mathbf{v}, b) \mapsto
     \begin{pmatrix}
         1 & \mathbf{k} & b \\
         \boldsymbol{0} & \mathbf{M} & \mathbf{v} \\
         0 & \boldsymbol{0} & 1
     \end{pmatrix}
 \end{equation}
The reader can check that multiplying matrices of the form above recovers the multiplication of $\mathrm{Osc}(V)$ in \autoref{eq:OscGpOper}.
At a deeper level, this observation a manifestation of  the fact that \(\mathsf{Osc}(V)\) is a contraction of \(\mathsf{Gl}(V \oplus \mathbb{R})\). \cite{gilmore2006lie}
\end{rmk}

\subsection{Annihilation subgroups}
 
The name ``oscillator algebra" comes from the observation that this Lie subalgebra appears in the Lie algebraic reformulation of Dirac's algebraic solution of the quantum mechanical harmonic oscillator \cite{dirac1981principles, gilmore2006lie}. This observation is relevant for our purposes because free quantum field theories are equivalent to harmonic oscillators \cite{zee2010quantum}, and hence one may expect Lie groups that act upon them to play a role in operations involving the free field term in the action.

In more detail, Dirac specifies the ground state of a harmonic oscillator as the unique vector of unit length which lies in the kernel of every annihilation operator.  (Recall that in the Schr\"odinger representation, this ground state is a Gaussian.)  The annihilator and creation operators generate the subgroup \(\mathsf{Heis}(V)\) of \(\mathsf{Osc}(V)\).  We will now characterize the way that the annihilation operators lie inside of \(\mathsf{Heis}(V)\) as sections of the canonical projection map $\mathsf{proj}:\mathsf{Heis}(V)=V^*\ltimes\mathsf{ILin}(V)\to V^*$. Note that $\mathsf{proj}$ is a group homomorphism. In a later section, after introducing an action of \(\mathsf{Osc}(V)\) on \(C^{\infty}(V)\), we will use this characterization to specify Gaussians as special elements of \(C^{\infty}(V)\).\\

\begin{dfn}
    For every symmetric tensor \(\mathbf{C} \in \Sym (V)\), define
    \begin{align}
        \mathit{An}(\mathbf{C}) \colon V^* &\to \mathsf{Osc} (V) \\
        \mathbf{k} &\mapsto \left(\mathbf{id},  \mathbf{k}, \mathbf{Ck}, \tfrac{1}{2} \mathbf{kCk} \right)
    \end{align}
\end{dfn}

\begin{lem}
$An(\mathbf{C})$ is a section (in the category of groups) of the projection $\mathsf{proj}$ above:
\begin{equation}
\begin{tikzcd}
    \mathsf{Heis}(V) \ar[r, "\mathsf{proj}"'] & V^* \ar[l, bend right, "An(\mathbf{C})"']
\end{tikzcd}
\end{equation}
\end{lem}

\begin{proof}
Using the bilinearity and symmetry of $\mathbf{C}$, we see:
\begin{align}
An(\mathbf{C})(\mathbf{k}_1+\mathbf{k}_2) &= (\id, \mathbf{k}_1+\mathbf{k}_2,\mathbf{C}(\mathbf{k}_1+\mathbf{k}_2),\frac{1}{2}(\mathbf{k}_1+\mathbf{k}_2)\mathbf{C}(\mathbf{k}_1+\mathbf{k}_2))\\
&=(\id ,\mathbf{k}_1+\mathbf{k}_2,\mathbf{Ck}_1+\mathbf{Ck}_2,\frac{1}{2}\mathbf{k_1}\mathbf{Ck}_1+\frac{1}{2}\mathbf{k}_2\mathbf{Ck}_2+\mathbf{k}_1\mathbf{Ck}_2)\\
&=(\id,\mathbf{k}_1,\mathbf{Ck}_1,\frac{1}{2}\mathbf{k}_1\mathbf{Ck}_1)\bullet(\id,\mathbf{k}_2,\mathbf{Ck}_2,\frac{1}{2}\mathbf{k}_2\mathbf{Ck}_2)\\
&= An(\mathbf{C})(\mathbf{k}_1)\bullet An(\mathbf{C})(\mathbf{k}_2)
\end{align}
Also note that $An(\mathbf{C})(0) = (\id,0,0,0)$, hence $An(\mathbf{C})$ preserves the identity and is a group homomorphism.
\end{proof}

\begin{cor}
    For all \(\mathbf{C} \in \Sym\), the image of \(\An(\mathbf{C})\) is a commutative subgroup of \(\mathsf{Heis}(V)\).
\end{cor}

\begin{proof}
    Since $V^*$ is a commutative group and \(\An (\mathbf{C})\) is a group homomorphism, its image is a commutative subgroup of \(\mathsf{Heis}(V)\).
\end{proof}

We can further note that:\\

\begin{lem}
    A map
    \begin{align}
     V^{*} &\to \mathsf{Heis}(V) \subset \mathsf{Osc}(V) \\
     \mathbf{k} &\mapsto (\id, \mathbf{k}, a (\mathbf{k}), b(\mathbf{k}))
 \end{align}  
 is a section of \(\mathrm{proj}\) iff
 \begin{align}
     a(0) &= 0 & a (\mathbf{k}_1 + \mathbf{k}_2) &= a (\mathbf{k}_1) + a (\mathbf{k}_2) \label{eq:aconds} \\
      b(0) &= 0 &  b (\mathbf{k}_1 + \mathbf{k}_2) &= b (\mathbf{k}_1) + b (\mathbf{k}_2) + \mathbf{k}_1 (a (\mathbf{k}_2)) \label{eq:bconds}
 \end{align}
 for all \(\mathbf{k}_1. \mathbf{k}_2 \in V^*\).
\end{lem}

\begin{proof}
    In order for the map to be a group homomorphism, we must have
    \begin{align}
        (\id, 0, a (0), b(0)) &= (\id, 0, 0, 0) \\
        (\id, \mathbf{k}_1 + \mathbf{k}_2, a (\mathbf{k}_1 + \mathbf{k}_2), b(\mathbf{k}_1 + \mathbf{k}_2)) &= (\id, \mathbf{k}_1, a (\mathbf{k}_1), b(\mathbf{k}_1)) \bullet (\id, \mathbf{k_2}, a (\mathbf{k_2}), b(\mathbf{k}_2))
    \end{align}
    The former condition will be satisfied iff \(a(0) = 0\) and \(b(0) = 0\).  Using \autoref{eq:OscGpOper}, we see that
    \begin{multline}
        (\id, \mathbf{k}_1, a(\mathbf{k}_1), b (\mathbf{k}_1)) \bullet (\id, \mathbf{k_2}, a(\mathbf{k_2}), b(\mathbf{k}_2)) \\ = (\id, \mathbf{k}_1 + \mathbf{k}_2, a(\mathbf{k}_1) + a(\mathbf{k}_2), b (\mathbf{k}_1) + b (\mathbf{k}_2) + \mathbf{k}_1 (a(\mathbf{k}_2))).
    \end{multline}
    Hence, the latter condition will be satisfied iff
    \begin{align}
        a (\mathbf{k}_1 + \mathbf{k}_2) &= a (\mathbf{k}_1) + a (\mathbf{k}_2) \\
        b (\mathbf{k}_1 + \mathbf{k}_2) &= b (\mathbf{k}_1) + b (\mathbf{k}_2) + \mathbf{k}_1 (a(\mathbf{k}_2)) 
    \end{align}
\end{proof}

\begin{dfn} \label{sumofsectionsdfn}
    If \(s_1, s_2 \colon V^{*} \to V^* \ltimes \mathsf{Lin} (V)\) are sections of \(\mathrm{proj}\) given as
    \begin{align}
        s_1 (\mathbf{k}) &= (\id, \mathbf{k}, a_1 (\mathbf{k}), b_1 (\mathbf{k})) \\
        s_2 (\mathbf{k}) &= (\id, \mathbf{k}, a_2 (\mathbf{k}), b_2 (\mathbf{k})),
    \end{align}
    define the map \([s_1 + s_2] \colon V^{*} \to V^* \ltimes \mathsf{Lin} (V)\) by the equation
    \begin{equation}
        [s_1 + s_2] (\mathbf{k}) = (\id, \mathbf{k}, a_1 (\mathbf{k}) + a_2 (\mathbf{k}), b_1 (\mathbf{k}) + b_2 (\mathbf{k})).
    \end{equation}
\end{dfn}

\begin{cor}
    If \(s_1, s_2 \colon V^{*} \to V^* \ltimes \mathsf{Lin} (V)\) are section of \(\mathrm{proj}\), then \([s_1 + s_2]\) is also a section of \(\mathrm{proj}\). Furthermore,
    \(\Gamma(\mathrm{proj})\) becomes a group when equipped with the operation \([+]\) and the identity element \(\id_{\Gamma(\mathrm{proj})} (\mathbf{k}) = (\id, \mathbf{k}, 0, 0)\).\\
\end{cor}

\begin{proof}
    We check:
    \begin{alignat}{3}
        &\hskip -6pt  b_1(k_1+k_2)+b_2(k_1+k_2)\\
        &= b_1(k_1)+b_1(k_2)+k_1(a_1(k_2)) &&\text{By \autoref{eq:bconds} for }b_1\\
        &+b_2(k_1)+b_2(k_2)+k_1(a_2(k_2)) &&\text{By \autoref{eq:bconds} for }b_2\\
        &=b_1(k_1)+b_2(k_1)+b_2(k_1)+b_2(k_2) \qquad &&\text{Rearranging}\\
        &+ k_1(a_1(k_2)+a_2(k_2)) &&\text{By linearity of }k_1
    \end{alignat}
    Similarly, an easy computation using \autoref{eq:aconds} can verify:
    \begin{equation}
        (a_1+a_2)(k_1+k_2)=a_1(k_1+k_2)+a_2(k_1+k_2)
    \end{equation}
    The operation $[+]$ is clearly associative and preserves the identity.
\end{proof}

\begin{lem}\label{lem:HomImg}
The map \(\An \colon \Sym (V)\to \Gamma(\mathsf{proj})\)
is a group isomorphism onto its image.
\end{lem}

\begin{proof}
We'll check that $\An$ is a group homomorphism by checking
\begin{equation} \label{sumofsectionsequality}
An(\mathbf{C}_1+\mathbf{C}_2) (\mathbf{k})= [An(\mathbf{C}_1) + An(\mathbf{C}_2)](\mathbf{k}) 
\end{equation}
for all \(\mathbf{k} \in V^*\). We compute:
\begin{align}
An(\mathbf{C}_1+\mathbf{C}_2)(\mathbf{k}) &= (\id,\mathbf{k},(\mathbf{C}_1+\mathbf{C}_2)\mathbf{k},\frac{1}{2}\mathbf{k}(\mathbf{C}_1+\mathbf{C}_2)\mathbf{k})\\
&=(\id,\mathbf{k},\mathbf{C}_1\mathbf{k}+\mathbf{C}_2\mathbf{k},\frac{1}{2}\mathbf{k}\mathbf{C}_1\mathbf{k}+\frac{1}{2}\mathbf{k}\mathbf{C}_2\mathbf{k})\\
&=[An(\mathbf{C}_1)+An(\mathbf{C}_2)](\mathbf{k})
\end{align}

It is trivial to note that $\mathbf{C}_1\neq\mathbf{C}_2$ iff there exists a $\mathbf{k}\in V^*$ with $\mathbf{C}_1\mathbf{k}\neq\mathbf{C}_2\mathbf{k}$, hence $An$ is injective and therefore an isomorphism onto its image.
\end{proof}

We now define three actions of \(\mathsf{GL}(V)\):\\
\begin{dfn} \label{def:ActSymSec}
    \begin{align}
        \ActSym \colon GL(V) \times \Sym (V) &\to \Sym (V) \\
        \ActSym (\mathbf{M}, \mathbf{C}) &= (\mathbf{M} \otimes \mathbf{M}) \circ \mathbf{C} \\
        \ActSec \colon GL(V) \times \Gamma(\mathrm{proj}) &\to \Gamma(\mathrm{proj}) \\
        \ActSec (\mathbf{M}, s)(\mathbf{k}) &= (\mathbf{M}, 0, 0, 0) \bullet s(\mathbf{kM}) \bullet (\mathbf{M}^{-1}, 0, 0, 0) \\
        \ActFun : GL(V) \times C^{\infty} (V) &\to C^{\infty} (V) \label{def:ActFun}\\
    \mathbf{M}, f &\mapsto \det(\mathbf{M})f \circ \mathbf{M}
    \end{align}
\end{dfn}

\begin{lem}\label{lem:ActSecAutomorphism}
    For all \(\mathbf{M} \in \mathsf{GL}(V)\), the map \(\ActSec (\mathbf{M}, -)  \colon \Gamma(\mathrm{proj}) \to \Gamma(\mathrm{proj})\) is a group automorphism.
\end{lem}

\begin{proof}
    From the definition of \(\ActSec\), it immediately follows that \(\ActSec (\mathbf{M}^{-1}, -)\) is the inverse of \(\ActSec (\mathbf{M}, -)\) so \(\ActSec (\mathbf{M}, -)\) is a bijection.
    
    We note that
    \begin{align}
        &\hskip -6 pt (\mathbf{M}, 0, 0, 0) \bullet (\id, \mathbf{kM}, \mathbf{a}, b) \bullet (\mathbf{M}^{-1}, 0, 0, 0) \\
        &= (\mathbf{M}, 0, 0, 0) \bullet (\mathbf{M}^{-1}, \mathbf{k}, \mathbf{a}, b) \\
        &= (\id, \mathbf{k}, \mathbf{Ma}, b).
    \end{align}
    so, given a section \(s (\mathbf{k}) = (\id, \mathbf{k}, \mathbf{a} (\mathbf{k}),  b (\mathbf{k}))\),
    we have
    \begin{equation}\label{eq:ActSect}
        \ActSec (\mathbf{M}, s) (\mathbf{k}) = (\id, \mathbf{k}, \mathbf{Ma} (\mathbf{kM}),  b (\mathbf{kM})).
    \end{equation}

    Using this fact, we check that \(\ActSec (\mathbf{M}, -)\) is group homomorphism:
    \begin{align}
        &\hskip -6 pt \ActSec (\mathbf{M}, \id_{\Gamma(\mathrm{proj})})(\mathbf{k}) \\
        &=  (\id, \mathbf{k}, 0, 0) \\
        &= \id_{\Gamma(\mathrm{proj})}(\mathbf{k}) \\
        &\hskip -6 pt \ActSec (\mathbf{M}, [s_1 + s_2])(\mathbf{k}) \\ &= (\id, \mathbf{k}, \mathbf{M}(\mathbf{a}_1 (\mathbf{kM}) + \mathbf{a}_2 (\mathbf{kM})),  b_1 (\mathbf{kM}) + b_2(\mathbf{kM})) \\ &= (\id, \mathbf{k}, \mathbf{Ma}_1 (\mathbf{kM}) + \mathbf{Ma}_2 (\mathbf{kM}),  b_1 (\mathbf{kM}) + b_2 (\mathbf{kM})) \\
        &= [\ActSec (\mathbf{M}, s_1) + \ActSec (\mathbf{M}, s_2)](\mathbf{k})
    \end{align}
\end{proof}

\begin{lem}\label{lem:conjugation}
     For all \(\mathbf{M} \in \mathsf{GL}(V)\) and all \(\mathbf{C} \in \Sym (V)\), we have \(\An (\ActSym (\mathbf{M}, \mathbf{C})) = \ActSec (\mathbf{M}, \An (\mathbf{C}))\).
\end{lem}

\begin{proof}
    Using the fact that \(\An (\mathbf{C}) = (\mathbf{id}, \mathbf{k}, \mathbf{Ck}, \tfrac{1}{2} \mathbf{kCk})\), we compute:
    \begin{alignat}{3}
        &\hskip -6pt \ActSec (\mathbf{M}, \An (\mathbf{C}))(\mathbf{k}) \\
        &= (\mathbf{id}, \mathbf{k}, \mathbf{MC}(\mathbf{kM}), \tfrac{1}{2} (\mathbf{kM})\mathbf{C}(\mathbf{kM})) &&\text{By \autoref{eq:ActSect}} \\
        &= (\mathbf{id}, \mathbf{k}, \ActSym (\mathbf{M}, \mathbf{C}) \mathbf{k}, \tfrac{1}{2} \mathbf{k} \ActSym (\mathbf{M}, \mathbf{C})\mathbf{k}) \qquad &&\text{Definition of }\ActSym \\
        &= \An (\ActSym (\mathbf{M}, \mathbf{C})) (\mathbf{k}) &&\text{By \autoref{eq:ActSect}}
    \end{alignat}    
\end{proof}

Hence we see that, for $\mathbf{C},\mathbf{C}'\in\Sym(V)$, we have $\mathrm{An}(\mathbf{C})$ conjugate to $\mathrm{An}(\mathbf{C}')$ in $\mathsf{Osc}(V)$.
We can summarize the foregoing result with the commutative diagram
\begin{equation} \label{eq:naturality}
    \xymatrix @C 42 pt {
    \Sym (V) \ar[r]^{\An} 
    \ar[d]_{\ActSym (\mathbf{M}, -)} & \Gamma(\mathrm{proj}) \ar[d]^{\ActSec(\mathbf{M}, -)} \\
    \Sym (V) \ar[r]^{\An} & \Gamma(\mathrm{proj})
    }
\end{equation}
which shows that \(\An\) is a homomrphism of actions.\cite{wells1976some}
  Using \(\ActSym\), we can define a semidirect product \(\mathsf{GL}(V) \ltimes \Sym\) and using \(\ActSec\), we can define a semidirect product \(\mathsf{GL}(V) \ltimes \Gamma(\mathrm{proj})\). Recall that the horizontal arrows are group homomrphisms by \autoref{lem:HomImg} so this is a diagram in the category of groups.  Hence, the homomrphism of actions decribed by this diagram means that we have a group homomorphism
\begin{equation}
    \xymatrix @C 54 pt {
        \mathsf{GL}(V) \ltimes \Sym 
        \ar[r]^{\id \times \An}
        & \mathsf{GL}(V) \ltimes \Gamma(\mathrm{proj}).
        }
\end{equation}

\begin{rmk}
    If we think of the action $\mathrm{GL}(V)\curvearrowright \Sym(V)$ as a functor $\ActSym :B\mathrm{GL}(V) \to \mathrm{Grp}$ and similarly for $\ActSec$, then \autoref{eq:naturality} says that $An$ is a natural transformation of these functors.
\end{rmk}

\subsubsection{Action on functions}

Having defined the relevant groups and considered some of their properties, we now act with our group $\mathsf{Osc}(V)$ on functions.\\

\begin{dfn}
    Define the map:
    \begin{align}
        \sigma \colon C^{\infty}(V) \times \mathsf{Osc}(V) &\to C^{\infty}(V) \\
        \sigma (f, (\mathbf{M}, \mathbf{k},  \mathbf{v},  c)) (\mathbf{x})
        &= e^{\mathbf{k} (\mathbf{x}) + c} f(\mathbf{M} \mathbf{x} + \mathbf{v})
    \end{align}
\end{dfn}

\begin{lem}
    The map \(\sigma\) defined above is a group homomorphism.
\end{lem}

\begin{proof}
    We check the defining conditions.  The identity behaves as it should, 
    \begin{equation}
        \sigma ((\mathbf{id}, 0, 0, 0), f)(\mathbf{x}) = f(\mathbf{x}).
    \end{equation}
    Acting twice and simplifying,
    \begin{align}
        \sigma \big( \sigma(f, (\mathbf{L}, \mathbf{j}, \mathbf{u}, a)),
        (\mathbf{M}, \mathbf{k}, \mathbf{v}, b) \big) 
        &= \sigma \big( e^{\mathbf{j} (\mathbf{x}) + a} f(\mathbf{L} \mathbf{x} + \mathbf{u}), (\mathbf{M}, \mathbf{k}, \mathbf{v}, b) \big) \\
        &= e^{\mathbf{k} \mathbf{x} + b} e^{\mathbf{j} (\mathbf{M} \mathbf{x} + \mathbf{v}) + a} \, f(\mathbf{L} (\mathbf{M} \mathbf{x} + \mathbf{v}) + \mathbf{u}) \\
        &= e^{(\mathbf{k} + \mathbf{jM})(\mathbf{x}) + a + b + \mathbf{j} (\mathbf{v})} \, f \big( \mathbf{LM} \mathbf{x} + \mathbf{L} \mathbf{v} + \mathbf{u} \big) \\
        &= \sigma \big( f, (\mathbf{LM}, \mathbf{jM} + \mathbf{k}, \mathbf{u} + \mathbf{L} \mathbf{v}, a + b + \mathbf{j} (\mathbf{v})) \big) \\
        &= \sigma \big( f,(\mathbf{L}, \mathbf{j}, \mathbf{u}, a) 
        \bullet (\mathbf{M}, \mathbf{k}, \mathbf{v}, b))
    \end{align}
    so \(\sigma\) respects the group operation.
\end{proof}

Next, we examine how the group \(\mathsf{Osc}(V)\) acts on convolutions. Recall the following definition:
\begin{dfn}
For $f,g\in L^1(\mathbb{R})$, we define the \textbf{convolution} of $f$ and $g$ to be the function:
\begin{equation}
(f*g)(x) := \int_{-\infty}^{\infty} f(y)g(x-y)\,dy
\end{equation}
\end{dfn}
Recall this operation is commutative and associative. Hence, $(C^\infty(V),*)$ is a commutative semigroup.

We constructed \(\mathsf{Osc}(V)\) as a semidirect product \(\mathsf{Osc}(V) = (\mathsf{GL}(V) \ltimes V^*) \ltimes (V \oplus \mathbb{R})\).  It turns out that each of the factors acts differently.\\

\begin{lem}\label{lem:ConvAct1}
    For all \(\mathbf{M} \in \mathsf{GL}(V)\) all \(\mathbf{k} \in V^*\), and all functions \(f,g\), we have
    \begin{equation}
        \sigma (f*g, (\mathbf{M}, \mathbf{k}, 0, 0))
        = \det (\mathbf{M}) \, \sigma (f, (\mathbf{M}, \mathbf{k}, 0, 0)) *
        \sigma (g, (\mathbf{M}, \mathbf{k}, 0, 0)).
    \end{equation}
\end{lem}

\begin{proof}
    \begin{align}
    \sigma (f*g, (\mathbf{M}, \mathbf{k}, 0, 0))
    &= e^{\mathbf{k}(\mathbf{x})} (f * g)(\mathbf{Mx}) \\
    &= e^{\mathbf{k}(\mathbf{x})} \int_V f(\mathbf{Mx}- \mathbf{y}) g(\mathbf{y}) \, d\mathbf{y} \\
    &= \det \mathbf{M} \, e^{\mathbf{k}(\mathbf{x})}
    \int_V f(\mathbf{Mx} - \mathbf{M}\mathbf{y}') \, g(\mathbf{My}') \, d\mathbf{y}' \\
     &= \det \mathbf{M}
    \int_V  e^{\mathbf{k}(\mathbf{x})- \mathbf{k}(\mathbf{y}')} f(\mathbf{Mx} - \mathbf{My}') \, e^{\mathbf{k}(\mathbf{y'})} g(\mathbf{My}') \, d\mathbf{y}' \\
    &= \det \mathbf{M}
    \int_V  e^{\mathbf{k}(\mathbf{x} - \mathbf{y}')} f(\mathbf{M} (\mathbf{x} - \mathbf{y}')) \, e^{\mathbf{k}(\mathbf{y'})} g(\mathbf{My}') \, d\mathbf{y}' \\
    &= \det \mathbf{M} \, \sigma (f, (\mathbf{M}, \mathbf{k}, 0, 0)) *
        \sigma (g, (\mathbf{M}, \mathbf{k}, 0, 0))
    \end{align}
    where \(\mathbf{y} = \mathbf{My}'\) so \(d\mathbf{y} =d(\mathbf{M}\mathbf{y}') = \det\mathbf{M}\,d\mathbf{y}'\)
\end{proof}

\begin{lem}\label{lem:ConvAct2}
    For all \(\mathbf{v} \in V\), all \(c \in \mathbb{R}\) and all functions \(f,g\), we have
    \begin{equation}
        \sigma (f*g, (\mathbf{id}, 0, \mathbf{v}, c))
        = \sigma (f. (\mathbf{id}, 0, \mathbf{v}, c)) * g
        = f * \sigma (g, (\mathbf{id}, 0, \mathbf{v}, c)).
    \end{equation}
\end{lem}

\begin{proof}
    \begin{align}
        \sigma (f*g, (\mathbf{id}, 0, \mathbf{v}, c))
        &= e^c (f * g) (\mathbf{x} + \mathbf{v}) \\
    &= \int_V e^c f(\mathbf{x} + \mathbf{v} - \mathbf{y}) \,
    g (\mathbf{y}) \, d \mathbf{y} 
    = \sigma (f. (\mathbf{id}, 0, \mathbf{v}, c)) * g \\
    &= \int_V f(\mathbf{x}  - \mathbf{y}) \,
    e^c g (\mathbf{y} + \mathbf{v}) \, d \mathbf{y}
    = f * \rho (g, (\mathbf{id}, 0, \mathbf{v}, c))
    \end{align}
\end{proof} 
\subsection{Group theoretic characterization of Gaussians}

Returning to our earlier comments, in the algebraic solution, the ground state of the harmonic oscillator is determined up to normalization by the condition that it lies in the kernel of all annihilation operators.  In the Schr{\"o}dinger representation, the ground state wavefunction is a Gaussian function.  We now translate these statements into the language of Lie algebras.\\

\begin{dfn}
    Recall, the \emph{normalized Gaussian distribution} with mean $\mu=0$ and covariance $\mathbf{C}\in\Symp (V)$ is given by:
    \begin{equation}
        \mathcal{N} [\mathbf{C}] ( \mathbf{v})
        = \frac{\exp \{-\tfrac{1}{2}
        \mathbf{v} \mathbf{C}^{-1} \mathbf{v}\}}
        {\int_{V} \exp \{-\tfrac{1}{2}
        \mathbf{v'} \mathbf{C}^{-1} \mathbf{v'}\}
        \, d\mathbf{v'}}
    \end{equation}
\end{dfn}

\pagebreak

\begin{thm} \label{thm:GaussChar}
    Equivalently, for $\mathbf{C}\in\Symp (V)$, a function $f$ satisfies:
    \begin{enumerate}
        \item $\sigma(f,\mathit{An}(\mathbf{C})(\mathbf{k}))=f$ for all $\mathbf{k} \in V^*$
        \item $f*1=1$
    \end{enumerate}
    iff $f=\mathcal{N}(\mathbf{C})$.
\end{thm}

\begin{proof}
    Unwinding the definitions, 
    \begin{align}
        &\hskip -6pt \sigma (f, \mathit{An}(\mathbf{C})(\mathbf{k}))(\mathbf{x}) \\ &= \sigma (f, (\id, \mathbf{k}, \mathbf{Ck}, \tfrac{1}{2} \mathbf{kCk}))(\mathbf{x}) \\
        &= e^{\mathbf{k} (\mathbf{x}) + \tfrac{1}{2} \mathbf{kCk}} \, f(\mathbf{x} + \mathbf{Ck})
    \end{align}
    for all \(\mathbf{k} \in V^*\). \\
    $(\Leftarrow)$ we verify that the Gaussian satisfies the first hypothesis:
    \begin{align}
        &\hskip -6pt \sigma (\mathcal{N}(\mathbf{C}), \mathit{An}(\mathbf{C})(\mathbf{k})) \\
        &= \sigma (\mathcal{N}(\mathbf{C})(0) \, \exp\{-\tfrac{1}{2} \mathbf{x} \mathbf{C}^{-1}\mathbf{x}\}, \mathit{An}(\mathbf{C})(\mathbf{k})) \\
        &= \mathcal{N}(\mathbf{C})(0) \, \exp \{-\tfrac{1}{2} (\mathbf{x} + \mathbf{Ck}) \mathbf{C}^{-1} (\mathbf{x} + \mathbf{Ck}) + \mathbf{k} (\mathbf{x}) + \tfrac{1}{2} \mathbf{kCk}\} \\
        &= \mathcal{N}(\mathbf{C})(0) \, \exp\{-\tfrac{1}{2} \mathbf{x} \mathbf{C}^{-1}\mathbf{x}\} \\
        &= \mathcal{N}(\mathbf{C})
    \end{align}
    It also satisfies the second hypothesis because \(\mathcal{N}(\mathbf{C}) * \boldsymbol{1} = \boldsymbol{1} \int_{V} \mathcal{N}(\mathbf{C}) (\mathbf{y}) \, d\mathbf{y} = \boldsymbol{1}\).\\
    $(\Rightarrow)$, if \(f = \sigma(f, \mathit{An}(\mathbf{C})(\mathbf{k}))\) for all \(\mathbf{k} \in V^*\), we have
    \begin{align}
        &\hskip -6pt \frac{f (\mathbf{x})}{\mathcal{N}(\mathbf{C}) (\mathbf{x})} \\
        &= \frac{\sigma(f, \mathit{An}(\mathbf{C})(\mathbf{k})) (\mathbf{x})}{\sigma(\mathcal{N}(\mathbf{C}), \mathit{An}(\mathbf{C})(\mathbf{k})) (\mathbf{x})} \\
        &= \frac{e^{\mathbf{k} (\mathbf{x}) + \tfrac{1}{2} \mathbf{kCk}} f (\mathbf{x} + \mathbf{Ck})}{e^{\mathbf{k} (\mathbf{x}) + \tfrac{1}{2} \mathbf{kCk}} \mathcal{N}(\mathbf{C}) (\mathbf{x}+\mathbf{Ck})} \\
        &= \frac{f (\mathbf{x} + \mathbf{Ck})}{\mathcal{N}(\mathbf{C}) (\mathbf{x} + \mathbf{Ck})}
    \end{align}
    Since \(\mathbf{C}\) is non-degenerate, for every \(\mathbf{x} \in V\) there exists a \(\mathbf{k} \in V^*\) such that \(\mathbf{x} + \mathbf{Ck} = 0\).  Hence, we have \(\mathcal{N}(\mathbf{C})(0) \, f(\mathbf{x}) = f(0) \,\mathcal{N}(\mathbf{C})(\mathbf{x})\).  Convolving both sides against the constant function, \(\mathcal{N}(\mathbf{C})(0) \, f * \boldsymbol{1} = f(0) \, \mathcal{N}(\mathbf{C}) * \boldsymbol{1}\).  Thus, \(f(0) = \mathcal{N}(\mathbf{C})(0)\), so \(f(\mathbf{x}) = \mathcal{N}(\mathbf{C})(\mathbf{x})\) for all \(\mathbf{x}\).
\end{proof}

Using this fact, we may deduce properties of Gaussian distributions.
By combining the pointwise addition of sections with the compatibility of convolution and the action of \(\mathsf{Osc}(V)\), we deduce that the convolution of two Gaussians is a Gaussian.\\

\begin{thm}\label{thm:GaussConv}
  \begin{equation}
    \mathcal{N}(\mathbf{A}) * \mathcal{N}(\mathbf{B}) =
    \mathcal{N}(\mathbf{A} + \mathbf{B})
  \end{equation}
\end{thm}

\begin{proof}
    For clarity, we introduce the abbreviations
    \begin{align}
        a_1 &= \mathbf{Ak} &
        a_2 &= \tfrac{1}{2} \mathbf{kAk} \\
        b_1 &= \mathbf{Bk} &
        b_2 &= \tfrac{1}{2} \mathbf{kBk}
    \end{align}
    so that \(\mathit{An}(\mathbf{A})(\mathbf{k}) = (\id, \mathbf{k}, a_1, a_2)\).  Using these abbreviations, we compute
    \begin{alignat}{3}\label{eq:assumption5}
        &\sigma (\mathcal{N}(\mathbf{A}) * \mathcal{N}(\mathbf{B}), (\id, 0, a_1 + b_1, a_2 + b_2)) &&\\
        &= \sigma (\mathcal{N}(\mathbf{A}) * \mathcal{N}(\mathbf{B}), (\id, 0, a_1, a_2) \cdot (\id, 0, b_1, b_2)) &&\text{Factor in }\mathsf{Heis}(V) \\
        &= \sigma (\sigma (\mathcal{N}(\mathbf{A}) * \mathcal{N}(\mathbf{B}), (\id, 0, a_1, a_2)), (\id, 0, b_1, b_2)) \qquad&&\text{Compatibility of }\sigma \\
        &= \sigma (\sigma (\mathcal{N}(\mathbf{A}), (\id, 0, a_1, a_2)) * \mathcal{N}(\mathbf{B}), (\id, 0, b_1, b_2)) \qquad &&\text{\autoref{lem:ConvAct2} with }\mathcal{N}(\mathbf{A}) \\
        &= \sigma (\mathcal{N}(\mathbf{A}), (\id, 0, a_1, a_2)) * \sigma (\mathcal{N}(\mathbf{B}), (\id, 0, b_1, b_2)) &&\text{\autoref{lem:ConvAct2} with }\mathcal{N}(\mathbf{B})
    \end{alignat}
    Building on this result, we show that \(\mathcal{N}(\mathbf{A}) * \mathcal{N}(\mathbf{B})\) is fixed by of \(\An(\mathbf{A} + \mathbf{B})(\mathbf{k})\) fo every every element \(\mathbf{k} \in V^*\):
    \begin{alignat}{3}
        &\sigma (\mathcal{N}(\mathbf{A}) * \mathcal{N}(\mathbf{B}), \An (\mathbf{A} + \mathbf{B})(\mathbf{k})) \\
        &= \sigma (\mathcal{N}(\mathbf{A}) * \mathcal{N}(\mathbf{B}), [\An (\mathbf{A}) + \mathit{An} (\mathbf{B})](\mathbf{k})) &&\text{By \autoref{sumofsectionsequality}}\\
        &= \sigma (\mathcal{N}(\mathbf{A}) * \mathcal{N}(\mathbf{B} , \left( \id, \mathbf{k}, a_1 + b_1, a_2 + b_2 \right) && \text{By \autoref{sumofsectionsdfn}}  \\
        &= \sigma (\mathcal{N}(\mathbf{A}) * \mathcal{N}(\mathbf{B}), \left( \id, 0, a_1 + b_1, a_2 + b_2 \right) \cdot (\id, \mathbf{k}, 0, 0)) &&\text{Factor in }\mathsf{Heis}(V) \\
        &= \sigma (\sigma (\mathcal{N}(\mathbf{A}) * \mathcal{N}(\mathbf{B}), \left( \id, 0, a_1 + b_1, a_2 + b_2 \right)), (\id, \mathbf{k}, 0, 0)) \qquad &&\text{Compatibility of }\sigma \\
        &= \sigma (\sigma (\mathcal{N}(\mathbf{A}), (\id, 0, a_1, a_2)) \\ &\qquad * \sigma (\mathcal{N}(\mathbf{B}), (\id, 0, b_1, b_2)), (\id, \mathbf{k}, 0, 0)) &&\text{By last calculation}\\
        &= \sigma (\sigma (\mathcal{N}(\mathbf{A}), (\id, 0, a_1, a_2)), (\id, \mathbf{k}, 0, 0)) \\ &\qquad* \sigma (\sigma (\mathcal{N}(\mathbf{B}), (\id, 0, b_1, b_2)), (\id, \mathbf{k}, 0, 0)) \quad &&\text{By \autoref{lem:ConvAct1}} \\ 
        &= \sigma (\mathcal{N}(\mathbf{A}), (\id, 0, a_1, a_2) \cdot (\id, \mathbf{k}, 0, 0)) \\ &\qquad * \sigma (\mathcal{N}(\mathbf{B}), (\id, 0, b_1, b_2) \cdot (\id, \mathbf{k}, 0, 0)) &&\text{Compatibility of }\sigma \nonumber \\
        &= \sigma (\mathcal{N}(\mathbf{A}), (\id, \mathbf{k}, a_1, a_2)) * \sigma (\mathcal{N}(\mathbf{B}), (\id, \mathbf{k}, b_1, b_2)) &&\text{Multiply in }\mathsf{Heis}(V) \\
        &= \sigma (\mathcal{N}(\mathbf{A}), \An (\mathbf{A})(\mathbf{k})) * \sigma (\mathcal{N}(\mathbf{B}), \An (\mathbf{B})(\mathbf{k})) && \text{Definition of }\An \\
        &= \mathcal{N}(\mathbf{A}) * \mathcal{N}(\mathbf{B}) &&\text{By \autoref{thm:GaussChar}}
    \end{alignat}
    Taking the convolution with a constant function, 
    \begin{alignat}{3}
        &(\mathcal{N}(\mathbf{A}) * \mathcal{N}(\mathbf{B})) * \boldsymbol{1} \\
        &\quad = \mathcal{N}(\mathbf{A}) * (\mathcal{N}(\mathbf{B}) * \boldsymbol{1}) \quad && \text{Associativity of }* \\
        &\quad = \mathcal{N}(\mathbf{A}) * \boldsymbol{1} &&\text{Normalization of }\mathcal{N}(\mathbf{B}) \\ 
        &\quad = \boldsymbol{1} &&\text{Normalization of }\mathcal{N}(\mathbf{A})
  \end{alignat}
  Recapping, we have shown that \(\sigma (\mathcal{N}(\mathbf{A}) * \mathcal{N}(\mathbf{B}), \An (\mathbf{A} + \mathbf{B})(\mathbf{k}) = \mathcal{N}(\mathbf{A}) * \mathcal{N}(\mathbf{B})\) and that \((\mathcal{N}(\mathbf{A}) * \mathcal{N}(\mathbf{B})) * \boldsymbol{1} = \boldsymbol{1}\).
  Hence, by \autoref{thm:GaussChar}, we conclude that \(\mathcal{N}(\mathbf{A}) *
  \mathcal{N}(\mathbf{B}) = \mathcal{N}(\mathbf{A} + \mathbf{B})\).
\end{proof}

\begin{lem}\label{lem:ActGaus}
  For any \(\mathbf{C} \in \Symp\) and any \(\mathbf{M} \in \mathsf{GL}(V)\), we have \(\ActFun (\mathcal{N}(\mathbf{C}), \mathbf{M}) = \mathcal{N}(\ActSym(\mathbf{M}^{-1},\mathbf{C}))\), or explicitly
  \begin{equation} \label{eq:ActGaus}
      \det(\mathbf{M}) \, \mathcal{N}(\mathbf{C}) \circ \mathbf{M} =\mathcal{N}((\mathbf{M}^{-1} \otimes \mathbf{M}^{-1}) \mathbf{C}).
  \end{equation}
\end{lem}

\begin{proof}
    For brevity, introduce the abbreviation \(\underline{\mathbf{M}} := (\mathbf{M}, 0 ,0 ,0))\).  We 
    now show that \(\sigma(\mathcal{N}(\mathbf{C}), \underline{\mathbf{M}}^{-1})\) is fixed by \(\An(\ActSym(\mathbf{M},\mathbf{C}))\).  For any \(\mathbf{k} \in V^*\), we compute:
    \begin{alignat}{3}
        &\hskip -6pt \sigma (\mathcal{N}(\mathbf{C}), \underline{\mathbf{M}}^{-1}) \\
        &= \sigma (\sigma (\mathcal{N}(\mathbf{C}), \An(\mathbf{C})(\mathbf{kM})), \underline{\mathbf{M}}^{-1}) &&\text{By \autoref{thm:GaussChar}} \\
        &= \sigma (\mathcal{N}(\mathbf{C}), \An(\mathbf{C})(\mathbf{kM}) \bullet \underline{\mathbf{M}}^{-1}) &&\sigma \text{ is a group action} \\
        &= \sigma (\mathcal{N}(\mathbf{C}), \underline{\mathbf{M}}^{-1} \bullet \underline{\mathbf{M}} \bullet \An(\mathbf{C})(\mathbf{kM}) \bullet \underline{\mathbf{M}}^{-1}) &&\text{Multiply by }\id = \underline{\mathbf{M}}^{-1} \bullet \underline{\mathbf{M}}\\
        &= \sigma ( \sigma (\mathcal{N}(\mathbf{C}), \underline{\mathbf{M}}^{-1}), \underline{\mathbf{M}} \bullet \An(\mathbf{C})(\mathbf{kM}) \bullet \underline{\mathbf{M}}^{-1})\qquad &&\sigma \text{ is a group action} \\
        &= \sigma ( \sigma (\mathcal{N}(\mathbf{C}), \underline{\mathbf{M}}^{-1}), \ActSec(\mathbf{M},\An(\mathbf{C}))(\mathbf{k})) &&\text{Definition of }\ActSec \\
        &= \sigma ( \sigma (\mathcal{N}(\mathbf{C}), \underline{\mathbf{M}}^{-1}), \An(\ActSym(\mathbf{M},\mathbf{C}))(\mathbf{k})) &&\text{By \autoref{lem:conjugation}}
    \end{alignat}

    We also note that, by \autoref{lem:ConvAct1}, we have
    \begin{align}
        \sigma (\mathcal{N}(\mathbf{C}) * 1, \underline{\mathbf{M}}) 
        &= \det (\mathbf{M}) \, \sigma (\mathcal{N}(\mathbf{C}), \underline{\mathbf{M}}) * \sigma (1, \underline{\mathbf{M}}) \\
        &= \det (\mathbf{M}) \, (\mathcal{N}(\mathbf{C}) \circ \mathbf{M}) * 1.
    \end{align}
    Since \(\mathcal{N}(\mathbf{C}) * \boldsymbol{1} = \boldsymbol{1}\) and \(\sigma (\boldsymbol{1}, \underline{\mathbf{M}}) = \boldsymbol{1}\), this means that \(\det (\mathbf{M}) \, (\mathcal{N}(\mathbf{C}) \circ \mathbf{M}) * \boldsymbol{1} = \boldsymbol{1}\).
    Hence, by \autoref{thm:GaussChar}, we have \(\det (\mathbf{M}) \, \mathcal{N}(\mathbf{C}) \circ \mathbf{M} = \mathcal{N}(\mathbf{C}')\).
\end{proof}

We can interpret \autoref{thm:GaussConv} as stating that \(\mathcal{N}\) is a homomorphism from \((\Symp (V), +)\) to \((C^{\infty} (V), *)\). 
We can interpret \autoref{lem:ActGaus} as saying that the two actions $\ActSym$ and $\ActFun$
fit into a commutative diagram:
\begin{equation}\label{eq:ActCom}
\begin{tikzcd}
\Symp(V) \ar[r, "\mathcal N"] \ar[d, "\ActSym(\mathbf{M}{,}-)"'] & C^\infty(V) \ar[d, "\ActFun(\mathbf{M}{,} -)"] \\
\Symp(V) \ar[r, "\mathcal N"] & C^\infty(V)
\end{tikzcd}
\end{equation}
As in \autoref{eq:naturality}, compatibility of $\mathcal{N}$ with the actions means that we have a group homomorphism between semidirect products,
\begin{equation}\label{eq:ActCom}
    \begin{tikzcd}
        \mathsf{GL}(V) \ltimes \Symp(V) \ar[r, "\id \times \mathcal N"]
        &\mathsf{GL}(V) \ltimes C^\infty(V).
    \end{tikzcd}
\end{equation}
In the next section, we will make use of this homomorphism to implement renormalization.
 
\section{Renormalization} \label{sec:renorm}

Having developed the necessary algebra, we now return to physics and carry out renormalization.  As described in the introduction on \autopageref{threeSteps}, this will be done in three steps: ``making blocks'', ``coarse graining'', and ``rescaling''.

To begin, we will rewrite the generating function of \autoref{def:generatingFunction} as a convolution.  Anticipating later developments, we will express the result in terms of the cumulant generating function rather than the moment generating function, but it is easy enough to pass between the two generating functions by exponentiating.

Recall, the propagator $\mathbf{P}\in V\otimes V$ is inverse to the kinetic term of the action $\mathbf{Q}\in V^*\otimes V^*$, there's an action $\mathbf{T} : \mathbb{R}^\times_{>0} \times V \to V$, and a source $\mathbf{J}\in V^*$.\\

\begin{thm} \label{GenAsConvo}
  \begin{align}
    W[\mathbf{P}, I](\mathbf{J})
    &= \tfrac{1}{2} \mathbf{J}\mathbf{P}\mathbf{J}  +
    \log (\mathcal{N}(\mathbf{P}) * \{\exp \circ I\} (\mathbf{P}\mathbf{J}))
  \end{align}
\end{thm}

\begin{proof}
  We begin by completing a square,
  \begin{align}
    &\hskip -6pt \tfrac{1}{2} \boldsymbol{\phi}
    \mathbf{Q} \boldsymbol{\phi}
    - \mathbf{J} (\boldsymbol{\phi}) \\
    &= \tfrac{1}{2} \boldsymbol{\phi}
    \mathbf{Q} \boldsymbol{\phi}
    - \mathbf{J} (\mathbf{P} \mathbf{Q}  \boldsymbol{\phi}) \\
    &= \tfrac{1}{2} (\boldsymbol{\phi} - \mathbf{J}\mathbf{P})
    \mathbf{Q} (\boldsymbol{\phi} - \mathbf{J}\mathbf{P})
    - \tfrac{1}{2}\mathbf{J}\mathbf{P}
    \mathbf{Q} (\mathbf{J}\mathbf{P}) \\
    &= \tfrac{1}{2} (\boldsymbol{\phi} - \mathbf{J}\mathbf{P})
    \mathbf{Q} (\boldsymbol{\phi} - \mathbf{J}\mathbf{P})
    - \tfrac{1}{2}\mathbf{J}\mathbf{P}\mathbf{J} .
  \end{align}
  Inserting the identity we just derived into the definition of the
  moment generating function, we recognize the result as
  convolution with a Gaussian distribution,
  \begin{align}
    Z[\mathbf{P}, I]
    &= \frac{\int_{V} \exp \{-\tfrac{1}{2} \boldsymbol{\phi}
    \mathbf{Q} \boldsymbol{\phi}
    + I (\boldsymbol{\phi}) + \mathbf{J}(\boldsymbol{\phi})\} \,
    \mathrm{d} \boldsymbol{\phi}}
    {\int_{V} \exp \{-\tfrac{1}{2} \boldsymbol{\phi}
    \mathbf{Q} \boldsymbol{\phi}
    \, \mathrm{d} \boldsymbol{\phi}\}} \\
    &= \exp \{\tfrac{1}{2} \mathbf{J}\mathbf{P}\mathbf{J} \}  \int_{V} \frac{\exp \{-\tfrac{1}{2}
    (\boldsymbol{\phi} - \mathbf{J}\mathbf{P})
    \mathbf{Q} (\boldsymbol{\phi} - \mathbf{P}\mathbf{J})\}}{\int_{V} \exp \{-\tfrac{1}{2} \boldsymbol{\phi}
    \mathbf{Q} \boldsymbol{\phi}
    \, \mathrm{d} \boldsymbol{\phi}\}}
    \exp \{I (\boldsymbol{\phi})\} \, d\boldsymbol{\phi} \\
    &= \exp \{\tfrac{1}{2} \mathbf{J}\mathbf{P}\mathbf{J} \} \,
    (\mathcal{N}(\mathbf{P}) * \{\exp \circ I\}) \, (\mathbf{P}\mathbf{J}).
  \end{align}
  Taking a logarithm, we obtain the corresponding identity for the
  cumulant generating function,
  \begin{align}
    W[\mathbf{P}, I](\mathbf{J})
    &= \tfrac{1}{2} \mathbf{J}\mathbf{P}\mathbf{J}  +
    \log (\mathcal{N}(\mathbf{P}) * \{\exp \circ I\}
    \, (\mathbf{P}\mathbf{J})).
  \end{align}
\end{proof} 

We recognize the term \(\tfrac{1}{2} \mathbf{J}\mathbf{P}\mathbf{J}\) as the cumulant generating function for the free field.
As we noted in the introduction, we may also write the formula as
\begin{align}\label{W-two-terms}
  W[\mathbf{P}, I](\mathbf{J})
  &= W[\mathbf{P}, 0](\mathbf{J}) +
  \tilde{W}[\mathbf{P}, I](\mathbf{P}\mathbf{J})
\end{align}
where we define \(\tilde{W}[\mathbf{P}, I] = \log \circ (\mathcal{N}(\mathbf{P}) * \{\exp \circ I\})\colon V \to \mathbb{R}\).

Returning to the three step procedure of renormalization, the first step of ``making blocks" is the simplest. 
Assume we've already employed some method of regularization to move from a continuum theory given by $(\mathbf{P},I)$ to a family of effective theories indexed by a length-scale $(\mathbf{P},I)_L$ with $L\in (0,\infty)$. 
Further assume the choice of regularization is equivariant in the sense that \(\mathbf{P}_{cL} = \ActSym (\mathbf{T}_c, \mathbf{P}_{L})\) for all \(c \in \mathbb{R}^{\times}_{>0}\) and monotonic in the sense that \(\mathbf{P}_{x} - \mathbf{P}_{y}\) is a positive quadratic form when $x<y$.
Let $L_0>0$ denote a fiducial length-scale which we can use to compare theories at different length-scales. 

To ``make blocks", we merely specify a longer length-scale $cL_0$ that we're renormalizing to (i.e. $c>1$).

To ``coarse grain", we can use \autoref{thm:GaussConv} to redo our cumulant generating function $W[\mathbf{P}_{L_0},I]$ in terms of a cut off propagator $\mathbf P_{cL_0}$ by altering the corresponding interaction term $I$ to $I_\text{eff}$ as follows:\\

\begin{lem}\label{eq:PIadj}
  For all \(\mathbf{P}_1, \mathbf{P}_2 \in \Symp (V)\)
   we have \( \tilde{W} [\mathbf{P}_1 + \mathbf{P}_2, I] = \tilde{W} [\mathbf{P}_1, \tilde{W} [\mathbf{P}_2, I]]\).
\end{lem}

\begin{proof}
  We compute
  \begin{alignat}{3}
    &\hskip -6pt \mathcal{N}(\mathbf{P}_1 + \mathbf{P}_2) * \{\exp \circ I\} \\
    &= (\mathcal{N}(\mathbf{P}_1)
    * \mathcal{N}(\mathbf{P}_2)) * \{\exp \circ I\} \quad &&\text{By \autoref{thm:GaussConv}} \\
    &= \mathcal{N}(\mathbf{P}_1)
    * (\mathcal{N}(\mathbf{P}_2) * \{\exp \circ I\}) &&\text{Associativity of convolution}
  \end{alignat}
  Inserting this result into the definition of \(\tilde{W}\) and
  simplifying,
  \begin{alignat}{3}
    &\hskip -6pt \tilde{W} [\mathbf{P}_1 + \mathbf{P}_2, I] \\ 
    &= \log (\mathcal{N}(\mathbf{P}_1 + \mathbf{P}_2)* \exp{I} \, ) &&\text{Definition of }\tilde{W}\\
    &= \log (\mathcal{N}(\mathbf{P}_1)
    * (\mathcal{N}(\mathbf{P}_2) * \{\exp \circ I\})) \quad &&\text{Above calculation}\\
    &= \tilde{W} [\mathbf{P}_1,
      \log (\mathcal{N}(\mathbf{P}_2) * \{\exp \circ I\})] &&\text{Definition of }\tilde{W}\\
    &= \tilde{W} [\mathbf{P}_1, \tilde{W} [\mathbf{P}_2, I]].&&\text{Definition of }\tilde{W}
  \end{alignat}
\end{proof}

\begin{cor}
    Taking $\mathbf{P}_{L_0} = \mathbf{P}_1+\mathbf{P}_2$ and $\mathbf{P}_{cL_0}=\mathbf{P}_2$ above (so that $\mathbf{P}_1 = \mathbf{P}_{L_0}-\mathbf{P}_{cL_0}$), 
    we can write the coarse graining from length-scale $L_0$ to $cL_0$ as:
    \begin{equation} \label{eq:coarseGrain}
        \tilde{W} [\mathbf{P}_{L_0}, I] = \tilde{W} [\mathbf{P}_{cL_0}, I_{\mathrm{eff}}]
    \end{equation}
\end{cor}

where \(I_{\text{eff}} = \tilde{W} [\mathbf{P}_{L_0} - \mathbf{P}_{cL_0}, I]\)\\

\begin{rmk}\label{rmk:CgAction}
     We can interpret \autoref{eq:PIadj} as saying $\tilde{W}$ is an action of the monoid $\Symp (V)$ on $C^\infty(V)$ given by:
     \begin{align}
         \Symp (V) \times C^\infty(V) &\to C^\infty(V)\\
         (\mathbf{P},I) &\mapsto \tilde{W}[\mathbf{P},I]
     \end{align}
\end{rmk}

Once we have coarse grained away all fluctuations between length-scales $L_0$ and $cL_0$, we have a generating function $\tilde{W}[\mathbf{P}_{cL_0},I_{\text{eff}}]$ that can give us expectation values at length-scale $cL_0$. 
The difference in expectation values between the scale $L_0$ theory and the scale $cL_0$ theory will result from two confounded factors: the coarse graining and the change in scale. To control for the change in scale, we finish the renormalization process by rescaling the theory back to scale $L_0$. 
This has the additional advantage that we will be able to express the renormalization transform as a semigroup action that only references the fiducial length-scale. 

To ``rescale", we'll use our transform $\mathbf{T}_c: V\to V$ (which abstracts the effect of rescaling spacetime on the space of fields) along with the following general lemma:\\

\begin{lem} \label{lem:rescaling}
  For all \(\mathbf{M} \in \mathsf{GL}(V)\) we have
  \begin{align}
    \tilde{W} [\mathbf{P}, I] \circ \mathbf{M}
    &= \tilde{W} [(\mathbf{M}^{-1}\otimes\mathbf{M}^{-1}) \mathbf{P}, I \circ \mathbf{M}] \label{eq:TactsonWtilde}  \\
    W [\mathbf{P}, I]
   \circ \mathbf{M}^{-1} &= W [(\mathbf{M}^{-1}\otimes\mathbf{M}^{-1})
   \, \mathbf{P}, I \circ \mathbf{M}].
  \end{align}
\end{lem}
  
\begin{proof}
  Acting on the convolution, we obtain
  \begin{alignat}{3}
    &(\mathcal{N}(\mathbf{P}) * \exp \circ I) \circ \mathbf{M} \\
    &= \sigma(\mathcal{N}(\mathbf{P})*\exp\circ I),(\mathbf{M},0,0,0)) &&\text{Definition of }\sigma\\
    &= \det \mathbf{M} \,
    (\mathcal{N}(\mathbf{P}) \circ \mathbf{M})
    * (\exp \circ I \circ \mathbf{M}) &&\text{By \autoref{lem:ConvAct1}}\\
    &= \mathcal{N}((\mathbf{M}^{-1}\otimes\mathbf{M}^{-1}) \mathbf{P})
    * (\exp \circ I \circ \mathbf{M})\qquad &&\text{By \autoref{eq:ActGaus}}
  \end{alignat}
  Taking the logarithm of both sides, \(\tilde{W} [\mathbf{P}, I]
  \circ \mathbf{M} = \tilde{W} [(\mathbf{M}^{-1}\otimes\mathbf{M}^{-1}) \mathbf{P}, I \circ \mathbf{M}]\).

    Next, we will compute how \(W\) transforms.  We will do this by considering how the two terms in \autoref{GenAsConvo} transform.  We begin by using some linear algebra to rewrite the first term:
  \begin{equation}
   W[\mathbf{P},0](\mathbf{JM}^{-1})=\frac{1}{2}(\mathbf{J}\mathbf{M}^{-1})
   \mathbf{P} (\mathbf{J}\mathbf{M}^{-1})
   =\frac{1}{2}\mathbf{J} ((\mathbf{M}^{-1}\otimes\mathbf{M}^{-1})
   \, \mathbf{P})  \mathbf{J}.
 \end{equation}
 Next we compute the second term:
 \begin{alignat}{3}
   &\tilde{W}[\mathbf{P}, I](\mathbf{P}(\mathbf{J}\mathbf{M}^{-1}) )\\
   &= \tilde{W}[\mathbf{P}, I] (\mathbf{M} \mathbf{M}^{-1}
   \mathbf{P}(\mathbf{J}\mathbf{M}^{-1})) &&\text{Compose by }\id=\mathbf{MM}^{-1} \\
   &= \tilde{W}[\mathbf{P},I]\circ \mathbf{M}(\mathbf{M}^{-1}\mathbf{P}(\mathbf{JM}^{-1}))&&\text{Rewrite by associativity}\\
   &= \tilde{W}[(\mathbf{M}^{-1}\otimes\mathbf{M}^{-1})
   \, \mathbf{P} ,
     I \circ \mathbf{M}] (\mathbf{M}^{-1}
   \mathbf{P}\mathbf{J} \mathbf{M}^{-1} ) &&\text{By \autoref{eq:TactsonWtilde}}\\
   &=\tilde{W}[(\mathbf{M}^{-1}\otimes\mathbf{M}^{-1})\, \mathbf{P} , I \circ \mathbf{M}](((\mathbf{M}^{-1}\otimes\mathbf{M}^{-1})\mathbf{P})\mathbf{J}) \qquad&&\text{By linear algebra}
 \end{alignat}
 Combining these two observations, we conclude that
 \begin{equation}
    W [\mathbf{P}, I]
   \circ \mathbf{M}^{-1} = W [(\mathbf{M}^{-1}\otimes\mathbf{M}^{-1})
   \, \mathbf{P}, I \circ \mathbf{M}].
 \end{equation}
\end{proof}

We can restate \autoref{eq:TactsonWtilde} as a diagram:
\begin{equation}
    \xymatrix @C 72 pt {
        C^\infty (V) \ar[d]_{\ActFun (\mathbf{M},-)} \ar[r]^{\tilde{W}(\mathbf{P}, -)} &  C^\infty (V) \ar[d]^{ \ActFun (\mathbf{M},-)} \\
        C^\infty (V) \ar[r]^{\tilde{W} (\ActSym (\mathbf{M}, \mathbf{P}), -)} &  C^\infty (V)}
\end{equation}  
 
Hence $\ActFun$ and \(\tilde{W}\) combine to give a monoid action that combines the coarse graining and rescaling transforms:
\begin{align}
    \mathop{\mathrm{Cgrl}} \colon \mathsf{GL}(V) \ltimes \Symp (V) &\to \mathop{\mathrm{End}_{\mathrm{Set}}} C^\infty (V) \label{eq:WAction} \\ 
    (\mathbf{M}, \mathbf{P}) &\mapsto (I \mapsto \ActFun (\mathbf{M}, \tilde{W} [\mathbf{P}, I]))
\end{align}

Using \(\mathbf{P}_{cL_0} = \ActSym (\mathbf{T}_c, \mathbf{P}_{L_0}) =  (\mathbf{T}_c \otimes \mathbf{T}_c) \mathbf{P}_{L_0}\), we can write:
\begin{alignat}{3}
    \tilde{W}[P_{L_0}, I] &= \tilde{W}[ (\mathbf{T}_c \otimes \mathbf{T}_c) \mathbf{P}_{L_0}, I_\text{eff}] &&\text{By \autoref{eq:coarseGrain}} \\
    &= \tilde{W}[(\mathbf{T}_c \otimes \mathbf{T}_c) \mathbf{P}_{L_0}, (I_\text{eff}\circ \mathbf{T}_c)\circ\mathbf{T}_c^{-1}] \qquad &&\text{Compose with } \id\\
    & = \tilde{W}[\mathbf{P}_{L_0}, I_\text{eff} \circ \mathbf{T}_c] \circ \mathbf{T}^{-1}_c \qquad&&\text{By \autoref{eq:TactsonWtilde}}
\end{alignat}
where \(I_\text{eff} = \tilde{W}[(\id - \mathbf{T}_c \otimes \mathbf{T}_c) \mathbf{P}_{L_0}, I]\) (or \(=\tilde{W}[\mathbf{P}_{L_0}-\mathbf{P}_{cL_0},I]\) as before).

We introduce some notation to express this more succinctly:\\
\begin{dfn}
    We define the \emph{renormalization transform} \(\mathcal{R}_c\) as
    \begin{align}
        \mathcal{R} \colon \mathbb{R}^\times_{\geq 1} &\to \mathrm{End}_\text{Set}(C^{\infty}(V)) \\
        c &\mapsto (I\mapsto \tilde{W}[(\id - \mathbf{T}_c \otimes \mathbf{T}_c) \mathbf{P}_{L_0}, I] \circ \mathbf{T} _c) \label{eq:renormAction}
    \end{align}
\end{dfn}

Using this notation, we can rewrite 
\begin{equation}
    \tilde{W}[\mathbf{P}_{L_0}, I] = \tilde{W}[\mathbf{P}_{L_0}, \mathcal{R}_c I] \circ \mathbf{T}^{-1}_c
\end{equation}

We will now factor the renormalization transform as composite of three maps.
\begin{equation}
    \mathbb{R}^\times_{\ge 1} \xrightarrow{\mathbf{T}} \mathsf{GL}(V) \xrightarrow{\mathop{\mathrm{Ur}} (\mathbf{P}_{L_0})} \mathsf{GL}(V) \ltimes \Symp (V) \xrightarrow{\mathop{\mathrm{Cgrl}}} \mathop{\mathrm{End}_{\mathrm{Set}}} (C^\infty (V))
\end{equation}

\begin{dfn}
    For each \(\mathbf{C} \in \Symp\), define
    \begin{align}
        \mathop{\mathrm{Ur}} (\mathbf{C}) \colon \mathsf{GL}(V) &\to \mathsf{GL} (V) \ltimes \Symp (V) \\
        \mathbf{M} & \mapsto (\mathbf{M}, \mathbf{C} - \ActSym (\mathbf{M}, \mathbf{C})).
    \end{align}
\end{dfn}

\pagebreak

\begin{lem}
    \(\mathrm{Ur}\) is a group homomorphism.
\end{lem}

\begin{proof}
    \begin{align}
        \mathop{\mathrm{Ur}} (\mathbf{C})(\id) &= (\id, \mathbf{C} - \ActSym (\id, \mathbf{C})) \\
        &= (\id, \mathbf{C} - \mathbf{C}) \\
        &= (\id, 0) \\
        &= \id
    \end{align}

    Given \(\mathbf{M}_1, \mathbf{M}_2 \in \mathsf{GL} (V)\), we compute:
    \begin{alignat}{3}
        &\hskip -6 pt \mathop{\mathrm{Ur}} (\mathbf{C})(\mathbf{M}_1) \bullet \mathop{\mathrm{Ur}} (\mathbf{C})(\mathbf{M}_2) \\
        &= (\mathbf{M}_1, \mathbf{C} - \ActSym (\mathbf{M}_1, \mathbf{C})) \bullet (\mathbf{M}_2, \mathbf{C} - \ActSym (\mathbf{M}_2, \mathbf{C})) \\
        &= (\mathbf{M}_1 \mathbf{M}_2, \mathbf{C} - \ActSym (\mathbf{M}_1, \mathbf{C}) + \ActSym (\mathbf{M}_1, \mathbf{C} - \ActSym (\mathbf{M}_2, \mathbf{C}))) \\
        &= (\mathbf{M}_1 \mathbf{M}_2, \mathbf{C} - \ActSym (\mathbf{M}_1, \mathbf{C}) + \ActSym (\mathbf{M}_1, \mathbf{C}) - \ActSym (\mathbf{M}_1,\ActSym (\mathbf{M}_2, \mathbf{C}))) \\
        &= (\mathbf{M}_1 \mathbf{M}_2, \mathbf{C}
        - \ActSym (\mathbf{M}_1,\ActSym (\mathbf{M}_2, \mathbf{C}))) \\
        &= (\mathbf{M}_1 \mathbf{M}_2, \mathbf{C}
        - \ActSym (\mathbf{M}_1 \mathbf{M}_2, \mathbf{C})) \\
        &= \mathop{\mathrm{Ur}} (\mathbf{C})(\mathbf{M}_1 \mathbf{M}_2)
    \end{alignat}
\end{proof}

\begin{rmk}
    Note that \(\mathop{\mathrm{Ur}} (\mathbf{C})\) is a section of the projection map \(\mathsf{GL} (V) \ltimes \Symp (V) \to \mathsf{GL}(V)\).  Furthermore, using the notation of \autoref{sumofsectionsdfn} for addition of sections, we have \([\mathop{\mathrm{Ur}} (\mathbf{C}_1) + \mathop{\mathrm{Ur}} (\mathbf{C}_2)] = \mathop{\mathrm{Ur}} (\mathbf{C}_1 + \mathbf{C}_2)\) by reasoning analogous to \autoref{lem:HomImg}, so \(\mathrm{Ur} \colon \Symp \to \Gamma(\mathsf{GL} (V) \ltimes \Symp (V) \to \mathsf{GL}(V))\) is a homomorphism.\\
\end{rmk}

\begin{lem}
The renormalization transform $\mathcal{R}$ defines a semigroup action.
\end{lem}
    
\begin{proof}
By \autoref{eq:TactsonWtilde}, we can rewrite the right-hand-side of \autoref{eq:renormAction} as:
\begin{equation}
    \mathcal{R}_c I
    = \tilde{W}[(\mathbf{T}_c^{-1} \otimes \mathbf{T}_c^{-1} - \id) \mathbf{P}_{L_0}, I \circ \mathbf{T}_c]
\end{equation}

We check that the purported action preserves products:
\begin{align}
    &\hskip -6pt \mathcal{R}_{c'} \mathcal{R}_{c} I \\
    &= \tilde{W}[(\mathbf{T}_{c'}^{-1} \otimes \mathbf{T}_{c'}^{-1} - \id) \mathbf{P}_{L_0}, \mathcal{R}_{c} I \circ \mathbf{T}_{c'}]\\
    &= \tilde{W}[(\mathbf{T}_{c'}^{-1} \otimes \mathbf{T}_{c'}^{-1} - \id) \mathbf{P}_{L_0}, \tilde{W}[(\mathbf{T}_c^{-1} \otimes \mathbf{T}_c^{-1} - \id) \mathbf{P}_{L_0}, I \circ \mathbf{T}_c] \circ \mathbf{T}_{c'}] \\
    &= \tilde{W}[(\mathbf{T}_{c'}^{-1} \otimes \mathbf{T}_{c'}^{-1} - \id) \mathbf{P}_{L_0}, \tilde{W}[(\mathbf{T}_{cc'}^{-1} \otimes \mathbf{T}_{cc'}^{-1} - \mathbf{T}_{c'}^{-1} \otimes \mathbf{T}_{c'}^{-1}) \mathbf{P}_{L_0}, I \circ \mathbf{T}_{cc'}]] \\
    &= \tilde{W}[(\mathbf{T}_{cc'}^{-1} \otimes \mathbf{T}_{cc'}^{-1} - \id) \mathbf{P}_{L_0}, I \circ \mathbf{T}_{cc'}]
\end{align}
\end{proof}

In \autoref{eq:WAction} we got an action \(\mathsf{GL} (V) \ltimes (\Symp (V), +)\) on \(C^{\infty} (V)\). Combining this with a rescaling of spacetime, we can describe renormalization as action of a one-parameter submonoid of \(\mathsf{GL} (V) \ltimes (\Symp (V), +)\) on \(C^{\infty} (V)\) specified by the composite:
\begin{equation} \label{eq:RenormSummary}
\begin{tikzcd}
    \mathbb{R}^\times_{\ge 1} \ar[r,"\mathbf{T}"] \ar[rrr, bend left, "\mathcal{R}"] & \mathsf{GL}(V) \ar[r, "\mathop{\mathrm{Ur}} (\mathbf{P}_{L_0})"] & \mathsf{GL}(V) \ltimes \Symp (V) \ar[r, "\mathop{\mathrm{Cgrl}}"] & \mathrm{End}_{\text{Set}}(C^\infty (V))\\
    c \ar[r, mapsto] & \mathbf{T}_c \ar[r, mapsto] & (\mathbf{T}_c, \mathbf{P}_{L_0} - \mathbf{P}_{cL_0}) \ar[r, mapsto] & (I\mapsto I_\text{eff} \circ \mathbf{T} _c)
\end{tikzcd}
\end{equation}
 
\section{Conclusions}\label{sec:conclusions}

To recap \autoref{Alg}, we defined the relevant groups $\mathrm{Heis}(V)$ and $\mathsf{Osc}(V)$ and showed some elementary facts about them. In \autoref{lem:HomImg}, we saw a homomorphism $An:\Symp(V) \to \Gamma(\mathrm{proj})$. 
We were able to define an action $\sigma$ of $\mathsf{Osc}(V)$ on $C^\infty(V)$, and in \autoref{thm:GaussChar} and \autoref{thm:GaussConv} we characterized $\mathcal{N}[\mathbf{C}]$ as the fixed point of $\sigma$ by elements in the image of $An(\mathbf{C})$, and further showed that $\mathcal{N}:\Symp(V)\to C^\infty(V)$ is a homomorphism.

Next, we showed that these constructions are compatible with the action of \(\mathsf{GL}(V)\) as summarized with the following diagram:
\begin{equation}
\begin{tikzcd}
\Symp(V) \ar[r, "\mathcal{N}"] \ar[d, "\mathsf{Act}_\mathrm{Sym}"'] & C^\infty(V) \ar[r, "\textrm{Conv}"] \ar[d, "\mathsf{Act}_{\mathrm{Fun}}"] & \mathrm{End}_\text{Set}(C^\infty(V)) \ar[d, "\ActFun\circ -"] \\
\Symp(V) \ar[r, "\mathcal{N}"] & C^\infty(V) \ar[r, "\textrm{Conv}"] & \mathrm{End}_\text{Set}(C^\infty(V))
\end{tikzcd}
\end{equation}

The vertical arrows $\ActSym$ and $\ActFun$ are defined in \autoref{def:ActSymSec}. The homomorphism \(\mathrm{Conv}\) consists of convolution with a smooth function.
Commutativity of the left square was shown in \autoref{lem:ActGaus}, commutativity of the right square is justified by \autoref{lem:ConvAct1} with $\mathbf{k}=0$. 
Composing across a row gives us the action noted in \autoref{rmk:CgAction}.
All of these results are used to justify why $\mathrm{Cgrl}$ in \autoref{eq:WAction} is a semigroup action $\mathrm{GL}(V)\ltimes \Symp(V)\curvearrowright \mathrm{End}_\text{Set}(C^\infty(V))$, which is used prominently in the renormalization procedure.

Let's recall the three steps of renormalization before matching the heuristic description with the algebra we've developed:

\begin{enumerate}
    \item We ``made blocks" by determining which energy scale $cL_0$ we're renormalizing to by writing \( P_{L_0} = (P_{L_0} - P_{cL_0}) + P_{cL_0} \).  
    \item We ``coarse grained" by using \autoref{eq:PIadj} to redo the interaction term as $I_{\text{eff}} = \tilde{W}[P_{L_0}-P_{cL_0},I]$, which can be used in the generating function:
    \begin{equation}
        \tilde{W}[P_{L_0},I](\mathbf{J}) = \tilde{W}[P_{c L_0},I_{\text{eff}}](\mathbf{J})
    \end{equation}
    \item We finish by rescaling using \autoref{lem:rescaling} and the transformation $\mathbf{T}_c:V\to V$ to get
    \begin{equation}
        W[P_{c L_0},I_{\text{eff}}](\mathbf{J}) = W[P_{L_0},\mathcal{R}_{c} I](\mathbf{J}\circ \mathbf{T}_c)
    \end{equation}
\end{enumerate}

We combine coarse graining and rescaling into the action $\mathrm{Cgrl}$ of \autoref{eq:WAction}, and incorporate making blocks by precomposing with $\mathrm{Ur}(\mathbf{P}_{L_0})\circ\mathbf{T}$ to specify the precise change in scale, resulting in the renormalization transform $\mathcal{R}$ of \autoref{eq:RenormSummary}.

Here we assumed the space of field configurations $V$ was finite-dimensional throughout. In future work, we will show how this generalizes to the infinite-dimensional case, in particular perturbative quantum field theory.  We will show how the differential equation corresponding to the action of the Lie algebra associated with $\mathrm{GL}(V)\ltimes \Symp(V)$ recovers the Polchinski flow equation, and we will validate the equivalence of this description with existing descriptions both in theory and by performing calculations in the $\phi^4$ theory.
 
\addcontentsline{toc}{section}{References}

\bibliographystyle{plain}
\bibliography{references}

\end{document}